\documentclass[onecolumn, draft, 12pt]{IEEEtran}
\usepackage{amsmath,cases}
\usepackage[mathscr]{euscript}
\usepackage{amsthm}
\usepackage{amssymb}
\usepackage{cite}
\usepackage{amsfonts}
\usepackage{enumerate}
\usepackage[final]{graphicx}
\usepackage{graphicx}
\usepackage{epstopdf}

\newtheorem{theorem}{Theorem}

\newtheorem{lemma}{Lemma}
\newtheorem{definition}{Definition}
\theoremstyle{remark}

\makeatletter

\newcommand{\Rmnum}[1]{\expandafter\@slowromancap\romannumeral #1@}
\makeatother

\begin{document}

\title{Laplace Functional Ordering of Point Processes in Large-scale Wireless Networks}
\author{Junghoon Lee and Cihan Tepedelenlio\u{g}lu, \emph{Senior Member, IEEE}
\thanks{This work was supported in part by the National Science Foundation under Grant CCF 1117041. This work was also supported in part by Institute for Information and communications Technology Promotion(IITP) grant funded by the Korea government(MSIT) (No. 2017-0-01973 (Korea-Japan) International collaboration of 5G mmWave based Wireless Channel Characteristic and Performance Evaluation in High Mobility Environments).}
\thanks{Junghoon Lee is with the Electronics and Telecommunications Research Institute (ETRI), Daejeon, South Korea (Email:jh.lee@etri.re.kr).}
\thanks{Cihan Tepedelenlio\u{g}lu is with the School of Electrical, Computer, and Energy Engineering, Arizona State University, Tempe, AZ 85287, USA (Email:cihan@asu.edu).} }
\maketitle

\begin{abstract}
Stochastic orders on point processes are partial orders which capture notions like being larger or more variable. Laplace functional ordering of point processes is a useful stochastic order for comparing spatial deployments of wireless networks. It is shown that the ordering of point processes is preserved under independent operations such as marking, thinning, clustering, superposition, and random translation. Laplace functional ordering can be used to establish comparisons of several performance metrics such as coverage probability, achievable rate, and resource allocation even when closed form expressions of such metrics are unavailable. Applications in several network scenarios are also provided where tradeoffs between coverage and interference as well as fairness and peakyness are studied. Monte-Carlo simulations are used to supplement our analytical results.

\end{abstract}


\section{Introduction}
\IEEEPARstart{P}{oint} processes have been used to describe spatial distribution of nodes in wireless networks. Examples include randomly distributed nodes in wireless sensor networks or ad-hoc networks \cite{Haenggi2008, Haenggi2009b, Kong2017} and the spatial distributions for base stations and mobile users in cellular networks \cite{Andrews2011, Mukherjee2012, Ye2013, Dhillon2017}. In the case of cognitive radio networks, locations of primary and secondary users have been modeled as point processes \cite{Rabbachin2011, Wen2012, Haenggi2012, Zhai2018}. Random translations of point processes hava been used for modeling of mobility of networks in \cite{Gong2013}. Stationary Poisson processes provide a tractable framework, but suffer from notorious modeling issues in matching real network distributions. Stochastic ordering of point processes provide an ideal framework for comparing two deployment/usage scenarios even in cases where the performance metrics cannot be computed in closed form. These partial orders capture intuitive notions like one point process being more dense, or more variable. Existing works on point process modeling for wireless networks have paid little attention to how two intractable scenarios can be nevertheless compared to aid in system optimization.

Recently stochastic ordering theory has been used for performance comparison in wireless networks which are modeled as point processes \cite{Blaszczyszyn2009, Blaszczyszyn2014, Madhusudhanan2012, Dhillon2013, Lee2014}. Directionally convex (DCX) ordering of point processes and its integral shot noise fields have been studied in \cite{Blaszczyszyn2009}. The work has been extended to the clustering comparison of point processes with various weaker tools including void probabilities and moment measures, than DCX ordering in \cite{Blaszczyszyn2014}. In \cite{Madhusudhanan2012}, usual stochastic ordering of random variables capturing carrier-to-interference ratio has been established in cellular systems. Ordering results for coverage probability and per user rate have been shown in multi-antenna heterogeneous cellular networks \cite{Dhillon2013}. In \cite{Lee2014}, Laplace functional (LF) ordering of point processes has been introduced and used to study interference distributions in wireless networks. Several examples of the LF ordering of specific point processes have been also introduced in \cite{Lee2014}, including stationary Poisson, mixed Poisson, Poisson cluster, and Binomial point processes.

In this paper, we apply the LF ordering concept to several general classes of point processes such as Cox, homogeneous independent cluster, perturbed lattice, and mixed binomial point processes which have been used to describe distributed nodes of wireless systems in the literature. We also investigate the preservation properties of the LF ordering of point processes with respect to independent operations such as marking, thinning, random translation, and superposition. We prove that the LF ordering of original point processes still holds after applying these operations on the point processes. To the best of our knowledge, there is no study of LF ordering of general classes of point processes and their preservation properties in the literature. Using these properties, we compare performances without having to obtain closed-form results for a wide range of performance metrics such as coverage probability, achievable rate, and resource allocation of different systems. In addition to the performance comparison, the stochastic ordering of point processes provides guidelines for system design such as network deployment and user selection schemes.

The paper is organized as follows: In Section \ref{sec:Math_Prelim}, we introduce mathematical preliminaries. Section \ref{sec:Ordering_of_Point_Processes} introduces ordering of point processes. In Section \ref{sec:Preserv_Ordering_of_Point_Processes}, we show the preservation properties of LF ordering. Section \ref{sec:App_to_Cellular_Networks} and \ref{sec:App_to_Cognitive_Networks} introduce applications of stochastic ordering of point processes in wireless networks. Section \ref{sec:Numerical_Results} presents simulations to corroborate our claims. Finally, the paper is summarized in Section \ref{sec:Summary}.

\section{Mathematical Preliminaries}
\label{sec:Math_Prelim}

\subsection{Stochastic Ordering of Random Variables}
Before introducing ordering of point processes, we briefly review some common stochastic orders between random variables, which can be found in \cite{Shaked, Shaked2}.

\subsubsection{Usual Stochastic Ordering}
\label{subsubsec:Usual_Stochastic_Ordering}
Let $X$ and $Y$ be two random variables (RVs) such that
\begin{equation}
\label{eqn:def_st_ordering}
P\left( X>x \right) \leq P\left( Y>x \right), -\infty < x < \infty.
\end{equation}
Then $X$ is said to be smaller than $Y$ in the \emph{usual stochastic order} (denoted by $X \leq_{\mathrm{st}} Y$). Roughly speaking, \eqref{eqn:def_st_ordering} says that $X$ is less likely than $Y$ to take on large values. To see the interpretation of this in the context of wireless communications, when $X$ and $Y$ are distributions of instantaneous SNRs due to fading, \eqref{eqn:def_st_ordering} is a comparison of outage probabilities. Since $X, Y$ are positive in this case, $x\in\mathbb{R}^+$ is sufficient in \eqref{eqn:def_st_ordering}.

\subsubsection{Laplace Transform Ordering}
\label{subsubsec:Laplace_Transform_Ordering}
Let $X$ and $Y$ be two non-negative RVs such that
\begin{equation}
\label{eqn:def_Lt_ordering}
\mathcal{L}_X(s)=\mathbb{E}[\exp{(-sX)}] \geq \mathbb{E}[\exp{(-sY)}]=\mathcal{L}_Y(s) \text{ for } s>0.
\end{equation}
Then $X$ is said to be smaller than $Y$ in the \emph{Laplace transform} (LT) \emph{order} (denoted by $X \leq_{\mathrm{Lt}} Y$). For example, when $X$ and $Y$ are the instantaneous SNR distributions of a fading channel, \eqref{eqn:def_Lt_ordering} can be interpreted as a comparison of average bit error rates for exponentially decaying instantaneous error rates (as in the case for differential-PSK (DPSK) modulation and Chernoff bounds for other modulations) \cite{Tepedelenlioglu2011}. The LT order $X \leq_{\mathrm{Lt}} Y$ is equivalent to
\begin{equation}
\label{eqn:def_Lt_ordering_conseq1}
\mathbb{E}[l(X)] \geq \mathbb{E}[l(Y)],
\end{equation}
for all \emph{completely monotonic} (c.m.) functions $l(\cdot)$ \cite[pp. 96]{Shaked2}. By definition, the derivatives of a c.m. function $l(x)$ alternate in sign: $(-1)^n \mathrm{d}^n l(x)/\mathrm{d}x^n \geq 0$, for $n=0,1,2,\dots$, and $x \geq 0$. An equivalent definition is that c.m. functions are positive mixtures of decaying exponentials \cite{Shaked2}. A similar result to \eqref{eqn:def_Lt_ordering_conseq1} with a reversal in the inequality states that
\begin{equation}
\label{eqn:def_Lt_ordering_conseq2}
X \leq_{\mathrm{Lt}} Y \Longleftrightarrow \mathbb{E}[l(X)] \leq \mathbb{E}[l(Y)],
\end{equation}
for all $l(\cdot)$ that have a completely monotonic derivative (c.m.d.). Finally, note that $X \leq_{\mathrm{st}} Y \Rightarrow$ \\ $X \leq_{\mathrm{Lt}} Y$. This can be shown by invoking the fact that $X \leq_{\mathrm{st}} Y$ is equivalent to $\mathbb{E}[l(X)] \leq \mathbb{E}[l(Y)]$ whenever $l(\cdot)$ is an increasing function \cite{Shaked2}, and that c.m.d. functions in \eqref{eqn:def_Lt_ordering_conseq2} are increasing.

\subsection{Point Processes and Random Measures}
\label{Point_Processes}
Point processes have been used to model large-scale networks \cite{Karr1991, Stoyan1995, Haenggi2008, Win2009, Haenggi2009b, Gulati2010, Haenggi2012, Haas2003, Garetto2011}. Since wireless nodes are usually not co-located, our focus is on \emph{simple} point processes, where only one point can exist at a given location. In addition, we assume the point processes are locally finite, i.e., there are finitely many points in any bounded set. Unlike \cite{Lee2014}, stationary and isotropic properties are not necessary in this paper. In what follows, we introduce some fundamental notions that will be useful.

\subsubsection{Campbell's Theorem}
\label{sec:Campbell_Theorem}
It is often necessary to evaluate the expected sum of a function evaluated at the point process $\Phi$. Campbell's theorem helps in evaluating such expectations. For any non-negative measurable function $u$ which runs over the set $\mathscr{U}$ of all non-negative functions on $\mathbb{R}^d$,
\begin{equation}
\label{eqn:Campbell_Theorem_PP}
\mathbb{E}\left[\sum_{x\in\Phi}u(x)\right]=\int_{\mathbb{R}^d}u(x)\Lambda(\mathrm{d}x) \text{.}
\end{equation}
The intensity measure $\Lambda$ of $\Phi$ in \eqref{eqn:Campbell_Theorem_PP} is a characteristic analogous to the mean of a real-valued random variable and defined as $\Lambda(B)=\mathbb{E}\left[\Phi(B)\right]$ for bounded subsets $B \subset \mathbb{R}^d$. So $\Lambda(B)$ is the mean number of points in $B$. If $\Phi$ is stationary then the intensity measure simplifies as $\Lambda(B)=\lambda \vert B \vert$ for some non-negative constant $\lambda$, which is called the intensity of $\Phi$, where $\vert B \vert$ denotes the $d$ dimensional volume of $B$. For stationary point processes, the right side in \eqref{eqn:Campbell_Theorem_PP} is equal to $\lambda\int_{\mathbb{R}^d}u(x)\mathrm{d}x$. Therefore, any two stationary point processes with same intensity lead to equal average sum of a function (when the mean value exists).

A random measure $\Psi$ is a function from Borel sets in $\mathbb{R}^d$ to random variables in $\mathbb{R}^{+}$. The \emph{Laplace functional} $L$ of random measure $\Psi$ is defined by the following formula
\begin{equation}
\label{eqn:def_Lf_of_PP}
L_{\Psi}(u) := \mathbb{E}\left[e^{-\int_{\mathbb{R}^d}u(x)\Psi(\mathrm{d}x)}\right] \text{.}
\end{equation}
The Laplace functional completely characterizes the distribution of the random measure \cite{Stoyan1995}. A point process $\Phi$ is a special case of a random measure $\Psi$ where the measure takes on values in the nonnegative integer random variables. In the case of the Laplace functional of a point process, $\int_{\mathbb{R}^d}u(x)\Phi(\mathrm{d}x)$ can be written as $\sum_{x\in\Phi}u(x)$ in \eqref{eqn:def_Lf_of_PP}. As an important example, the Laplace functional $L$ of Poisson point process of intensity measure $\Lambda$ is
\begin{equation}
\label{eqn:def_Lf_of_PPP}
L_{\Phi}(u) = \exp\left\{-\int_{\mathbb{R}^d}\left[1-\exp(-u(x))\right]\Lambda(\mathrm{d}x)\right\} \text{.}
\end{equation}
If the Poisson point process is stationary, the Laplace functional simplifies with $\Lambda(\mathrm{d}x)=\lambda\mathrm{d}x$.

\subsubsection{Laplace Functional Ordering}
In this section, we introduce the Laplace functional stochastic order between random measures which can also be used to order point processes.
\begin{definition}
\label{def_Lf_ordering}
Let $\Psi_1$ and $\Psi_2$ be two random measures such that
\begin{equation}
\label{eqn:def_Lf_ordering}
L_{\Psi_1}(u)=\mathbb{E}\left[e^{-\int_{\mathbb{R}^d}u(x)\Psi_1(\mathrm{d}x)}\right] \geq \mathbb{E}\left[e^{-\int_{\mathbb{R}^d}u(x)\Psi_2(\mathrm{d}x)}\right]=L_{\Psi_2}(u)
\end{equation}
where $u(\cdot)$ runs over the set $\mathscr{U}$ of all non-negative functions on $\mathbb{R}^d$. Then $\Psi_1$ is said to be smaller than $\Psi_2$ in the Laplace functional (LF) order (denoted by $\Psi_1 \leq_{\mathrm{Lf}} \Psi_2$).
\end{definition}
In this paper, we focus on the LF order of point processes unless otherwise specified. Note that the LT ordering in \eqref{eqn:def_Lt_ordering} is for RVs, whereas the LF ordering in \eqref{eqn:def_Lf_ordering} is for point processes or random measures. They can be connected in the following way:
\begin{equation}
\label{eqn:Lf_to Lt_ordering}
\Phi_1 \leq_{\mathrm{Lf}} \Phi_2 \Longleftrightarrow \sum_{x\in\Phi_1}u(x) \leq_{\mathrm{Lt}} \sum_{x\in\Phi_2}u(x) \text{,  } \forall u \in \mathscr{U} \text{.}
\end{equation}
Hence, it is possible to think of LF ordering of point processes as the LT ordering of their aggregate processes. Intuitively, the LF ordering of point processes can be interpreted as the LT ordering of their aggregate interferences. The LF ordering of point processes also can be translated into the ordering of coverage probabilities and spatial coverages which will be discussed in detail later.

\subsubsection{Voronoi Cell and Tessellation}
The Voronoi cell $V(x)$ of a point $x$ of a general point process $\Phi \subset \mathbb{R}^d$ consists of those locations of $\mathbb{R}^d$ whose distance to $x$ is not greater than their distance to any other point in $\Phi$, i.e.,
\begin{equation}
\label{eqn:def_Voronoi_cell}
V(x) := \lbrace y\in\mathbb{R}^d: \lVert x-y \rVert \leq \lVert z-y \rVert \text{ } \forall z\in\Phi \setminus {x} \rbrace \text{.}
\end{equation}
The Voronoi tessellation (or Voronoi diagram) is a decomposition of the space into the Voronoi cells of a general point process.

\section{Ordering of General Classes of Point Processes}
\label{sec:Ordering_of_Point_Processes}
The examples for LF orderings of some specific point processes have been provided in \cite{Lee2014}. In this section, we introduce the LF ordering of general classes of point processes.

\subsection{Cox Processes}
A generalization of the Poisson process is to allow for the intensity measure itself being random. The resulting process is then Poisson conditional on the intensity measure. Such processes are called \emph{doubly stochastic Poisson processes} or \emph{Cox processes}. Consider a random measure $\Psi$ on $\mathbb{R}^d$. Assume that for each realization $\Psi=\Lambda$, an independent Poisson point process $\Phi$ of intensity measure $\Lambda$ is given. The random measure $\Psi$ is called the driving measure for a Cox process. The LF ordering of Cox processes depends on their driving random measures.
\begin{theorem}
\label{Lf_ordering_of_Cox_PPs}
Let $\Phi_1$ and $\Phi_2$ be two Cox processes with driving random measures $\Psi_1$ and $\Psi_2$ respectively. If $\Psi_1 \leq_{\mathrm{Lf}} \Psi_2$, then $\Phi_1\leq_{\mathrm{Lf}} \Phi_2$.
\end{theorem}
\begin{proof}
The proof is given in Appendix \ref{app:Proof_for_Lf_Cox_PPs}.
\end{proof}
The mixed Poisson process is a simple instance of a Cox process, where the random measure $\Psi$ is described by a positive random constant $X$ so that $\Psi(B)=X|B|$. Since the Laplace functional of the mixed Poisson process can be expressed as $L_{\Phi}(u)=\mathbb{E}_{X}\left[\exp\{-X\int_{\mathbb{R}^d}[1-\exp(-u(x))]\mathrm{d}x\}\right]$, using \eqref{eqn:def_Lf_of_PPP}, and because $\int_{\mathbb{R}^d}[1-\exp(-u(x))]\mathrm{d}x \geq 0$ and the c.m. property of $\exp(-ax), a \geq 0$, the LF ordering of mixed Poisson processes has the following relationship: if $X_1 \leq_{\mathrm{Lt}} X_2$, then $\Phi_1\leq_{\mathrm{Lf}} \Phi_2$.

\subsection{Homogeneous Independent Cluster Processes}
\label{subsec:Ordering_of_Cluseter_PPs}
A general cluster process is generated by taking a parent point process and daughter point processes, one per parent, and translating the daughter processes to the position of their parent. The cluster process is then the union of all the daughter points. Denote the parent point process by $\Phi_{\mathrm{p}}=\lbrace x_1, x_2,\dots\rbrace$, and let $n\in\mathbb{N}\cup\{\infty\}$ be the number of parent points. Further let $\{\Phi_i\},i\in\mathbb{N}$, be a family of finite points sets, the untranslated clusters or daughter processes. The cluster process is then the union of the translated clusters:
\begin{equation}
\label{eqn:General_cluster_PP}
\Phi := \bigcup_{i=1}^{n}\Phi_i+x_i \text{.}
\end{equation}
If the parent process is a lattice, the process is called a \emph{lattice cluster process}. Analogously, if the parent process is a Poisson point process, the resulting process is a \emph{Poisson cluster process}.

If the parent process is stationary and the daughter processes $\{\Phi_i\}_i$ are finite point sets which are independent of each other and are independent of $\Phi_{\mathrm{p}}$, and have the same distribution, the procedure is called \emph{homogeneous independent clustering}. In this case, only the statistics of one cluster need to be specified, which is usually done by referring to the \emph{representative cluster}, denoted by $\Phi_0$ which is distributed the same as any $\Phi_i,i\in\mathbb{N}$. In this class of point processes, the LF ordering depends on the parent process $\Phi_{\mathrm{p}}$ and the representative process $\Phi_0$ as follows:
\begin{theorem}
\label{Lf_ordering_of_homo_indep_cluster_PPs}
Let $\Phi_1$ and $\Phi_2$ be two homogeneous independent cluster processes having representative clusters $\Phi_{0_1}$ and $\Phi_{0_2}$ respectively. Also, let $\Phi_{\mathrm{p}_1}$ and $\Phi_{\mathrm{p}_2}$ be the parent point processes of two homogeneous independent cluster processes $\Phi_1$ and $\Phi_2$ respectively. If $\Phi_{\mathrm{p}_1} \leq_{\mathrm{Lf}} \Phi_{\mathrm{p}_2}$ and $\Phi_{0_1} \leq_{\mathrm{Lf}} \Phi_{0_2}$, then $\Phi_1\leq_{\mathrm{Lf}} \Phi_2$.
\end{theorem}
\begin{proof}
The proof is given in Appendix \ref{app:Proof_for_homo_indep_cluster_PPs}.
\end{proof}

\subsection{Perturbed Lattice Processes with Replicating Points}
\label{subsec:Pert_lat_PPs}
Lattices are deterministic point processes defined as
\begin{equation}
\label{eqn:def_gen_lat_PP}
\mathbb{L} := \lbrace u\in\mathbb{Z}^d:\mathbf{G}u \rbrace \text{,}
\end{equation}
where $\mathbf{G}\in\mathbb{R}^{d \times d}$ is a matrix with $\det\mathbf{G}\neq0$, the so-called generator matrix. The volume of each Voronoi cell is $V=|\mathrm{det}\mathbf{G}|$ and the intensity of the lattice is $\lambda=1/V$ \cite{Haenggi2010}. The perturbed lattice process is a lattice cluster process. Denote the lattice point process by $\mathbb{L}=\lbrace x_1, x_2,\dots\rbrace$, and let $n\in\mathbb{N}\cup\{\infty\}$ be the number of lattice points. Further let $\{\Phi_i\},i\in\mathbb{N}$, be untranslated clusters. In each cluster, the number of daughter points are a random variable $X$, independent of each other, and identically distributed. Moreover, these points are distributed by some given spatial distribution. The entire process is then the union of the translated clusters as in \eqref{eqn:General_cluster_PP}. If the replicating points are uniformly distributed in the Voronoi cell of the original lattice, the resulting point process is a stationary point process and called a \emph{uniformly perturbed lattice process}. If, moreover, the number of replicas $X$ are Poisson random variables, the the resulting process is a stationary Poisson point process \cite{Blaszczyszyn2010}. Now, we can define the following LF ordering of such point processes.
\begin{theorem}
\label{Lf_ordering_of_pert_lat_PPs}
Let $\Phi_1$ and $\Phi_2$ be two uniformly perturbed lattice processes with numbers of replicas being non-negative integer valued random variables $X_1$ and $X_2$ respectively, and with the same mean $\mathbb{E}[X_1]=\mathbb{E}[X_2]=1$. If $X_1 \leq_{\mathrm{Lt}} X_2$, then $\Phi_1\leq_{\mathrm{Lf}} \Phi_2$.
\end{theorem}
\begin{proof}
The proof is given in Appendix \ref{app:Proof_for_Lf_pert_lat_PPs}.
\end{proof}
Based on Theorem 5.A.21 in \cite{Shaked}, the \emph{smallest} and \emph{biggest} LT ordered random variables can be defined as follows: Let $Y$ be a random variable such that $P\lbrace Y=0 \rbrace=1-P\lbrace Y=2 \rbrace=1/2$ and let $Z$ be a random variable degenerate at $1$. Let $X$ be a non-negative random variable with mean $1$. Then
\begin{equation}
\label{eqn:Lt_ordering_of_RVs}
Y \leq_{\mathrm{Lt}} X \leq_{\mathrm{Lt}} Z \text{.}
\end{equation}
From Theorem \ref{Lf_ordering_of_pert_lat_PPs} and \eqref{eqn:Lt_ordering_of_RVs}, the uniformly perturbed lattice processes with replicating points with non-negative integer valued distribution $Y$ and $Z$ will be the \emph{smallest} and \emph{biggest} LF ordered point processes respectively among uniformly perturbed lattice processes with the same average number of points. The smallest LF ordered uniformly perturbed lattice process exhibits clustering since some Voronoi cells have $2$ points but other cells do not have any point. This observation is in line with the intuition that clustering diminishes point processes in the LF order.

\subsection{Mixed Binomial Point Processes}
\label{subsec:Mixed_BPPs}
In binomial point processes, there are a total of fixed $N$ points uniformly distributed in a bounded set $B \in \mathbb{R}^d$. The density of the process is given by $\lambda=N/|B|$ where $|B|$ is the volume of $B$. If the number of points $N$ is random, the point process is called as a mixed binomial point process. As an example, with Poisson distributed $N$, the point process is called as a finite Poisson point process. The intensity measure of mixed binomial point processes is $\Lambda(B)=\lambda|B|$. In these point processes, one can show the following:
\begin{theorem}
\label{Lf_ordering_of_mixed_BPPs}
Let $\Phi_1$ and $\Phi_2$ be two mixed binomial point process with non-negative integer valued random distribution $N_1$ and $N_2$ respectively. If $N_1 \leq_{\mathrm{Lt}} N_2$, then $\Phi_1\leq_{\mathrm{Lf}} \Phi_2$.
\end{theorem}
\begin{proof}
The proof is given in Appendix \ref{app:Proof_for_Lf_mixed_BPPs}.
\end{proof}
Similar to Theorem \ref{Lf_ordering_of_pert_lat_PPs}, Theorem \ref{Lf_ordering_of_mixed_BPPs} enables LF ordering of two point processes whenever an associated discrete random variable is LT ordered.

\section{Preservation of Stochastic Ordering of Point Processes}
\label{sec:Preserv_Ordering_of_Point_Processes}
In what follows, we will show that the LF ordering between two point processes is preserved after applying independent operations on point processes such as marking, thinning, random translation, and superposition of point processes.

\subsection{Marking}
\label{subsec:independent_marking}
Consider the $d$ dimensional Euclidean space $\mathbb{R}^d$, $d \geq 1$, as the state space of the point process. Consider a second space $\mathbb{R}^{\ell}$, called the space of marks. A marked point process $\widetilde{\Phi}$ on $\mathbb{R}^d\times\mathbb{R}^{\ell}$ (with points in $\mathbb{R}^d$ and marks in $\mathbb{R}^{\ell}$) is a locally finite, random set of points on $\mathbb{R}^d$, with some random vector in $\mathbb{R}^{\ell}$ attached to each point. A marked point process is said to be independently marked if, given the locations of the points in $\Phi$, the marks are mutually independent random vectors in $\mathbb{R}^{\ell}$, and if the conditional distribution of the mark $m_x$ of a point $x \in \Phi$ depends only on the location of this point $x$ it is attached to.
\begin{lemma}
\label{Lf_ordering_marked_PP}
Let $\Phi_1$ and $\Phi_2$ be two point processes in $\mathbb{R}^d$. Also let $\widetilde{\Phi}_1$ and $\widetilde{\Phi}_2$ be independently marked point processes with marks $m_x$ with identical distribution in $\mathbb{R}^{\ell}$. If $\Phi_1 \leq_{\mathrm{Lf}} \Phi_2$, then $\widetilde{\Phi}_1 \leq_{\mathrm{Lf}} \widetilde{\Phi}_2$.
\end{lemma}
\begin{proof}
The proof is given in Appendix \ref{app:Proof_for_marked_PPs}.
\end{proof}

\subsection{Thinning}
A thinning operation uses a rule to delete points of a basic process $\Phi$, thus yielding the \emph{thinned point process} $\Phi_{\mathrm{th}}$, which can be considered as a subset of $\Phi$. The simplest thinning is \emph{$p$-thinning}: each point of $\Phi$ has probability $1-p$ of suffering deletion, and its deletion is independent of locations and possible deletions of any other points of $\Phi$. A natural generalization allows the retention probability $p$ to depend on the location $x$ of the point. A deterministic function $p(x)$ is given on $\mathbb{R}^d$, with $0 \leq p(x) \leq 1$. A point $x$ in $\Phi$ is deleted with probability $1-p(x)$ and again its deletion is independent of locations and possible deletions of any other points. The generalized operation is called \emph{$p(x)$-thinning}. In a further generalization the function $p$ is itself random. Formally, a random field $\boldsymbol{\pi} = \lbrace 0 \leq \pi(x) \leq 1:x\in\mathbb{R}^d \rbrace$ is given which is independent of $\Phi$. A realization $\varphi_{\mathrm{th}}$ of the thinned process $\Phi_{\mathrm{th}}$ is constructed by taking a realization $\varphi$ of $\Phi$ and applying $p(x)$-thinning to $\varphi$, using for $p(x)$ a sample $\lbrace p(x):x\in\mathbb{R}^d \rbrace$ of the random field $\boldsymbol{\pi}$. Given $\pi(x)=p(x)$ and given $\Phi=\varphi$, the probability of $x$ in $\Phi$ also belonging to $\Phi_{\mathrm{th}}$ is $p(x)$. As long as a independent thinning operation regardless of $p, p(x), \text{and } \pi(x)$ is applied on point processes, the LF ordering of the original pair of point processes is retained:
\begin{lemma}
\label{Lf_ordering_of_thinning_PPs}
Let $\Phi_1$ and $\Phi_2$ be two point processes in $\mathbb{R}^d$ and $\Phi_{\mathrm{th},1}$ and $\Phi_{\mathrm{th},2}$ be independently thinned point processes both with any identical independent thinning operation which could be either $p, p(x), \text{or } \pi(x)\text{-thinning}$ on both $\Phi_1$ and $\Phi_2$. If $\Phi_1 \leq_{\mathrm{Lf}} \Phi_2$, then $\Phi_{\mathrm{th},1} \leq_{\mathrm{Lf}} \Phi_{\mathrm{th},2}$.
\end{lemma}
\begin{proof}
The proof is given in Appendix \ref{app:Proof_for_thinning_PPs}.
\end{proof}
Since the thinned point process $\Phi_{\mathrm{th}}$ is a locally finite random set of points on $\mathbb{R}^d$, with a binary random variable in $\mathbb{R}^{+}$ attached to each point, independent thinning can be considered as the independent marking operation on a point process as discussed in the previous section.

\subsection{Random Translation}
\label{subsec:Random_Translation}
In this section, the stochastic operation that we consider is random translation. Each point $x$ in the realization of some initial point process $\Phi$ is shifted independently of its neighbors through a random vector $t_x$ in $\mathbb{R}^d$ where $\lbrace t_x \rbrace_{x}$ are independent each other and the conditional distribution of a random vector $t_x$ of a point $x \in \Phi$ depends only on the location of the point $x$. The resulting process is $\Phi_{\mathrm{rt}}:=\lbrace x+t_x:x\in\Phi \rbrace$. The random translation preserves the LF ordering of point process as follows:
\begin{lemma}
\label{Lf_ordering_translation_PP}
Let $\Phi_1$ and $\Phi_2$ be two point processes in $\mathbb{R}^d$ and $\Phi_{\mathrm{rt},1}$ and $\Phi_{\mathrm{rt},2}$ be the translated point processes with common distribution for the translation $t_x$. If $\Phi_1 \leq_{\mathrm{Lf}} \Phi_2$, then $\Phi_{\mathrm{rt},1} \leq_{\mathrm{Lf}} \Phi_{\mathrm{rt},2}$.
\end{lemma}
\begin{proof}
The proof is given in Appendix \ref{app:Proof_for_translation_PPs}.
\end{proof}
Similar to the independent thinning operation, since the random translated point process $\Phi_{\mathrm{rt}}$ is a locally finite random set of points on $\mathbb{R}^d$, with some random vector in $\mathbb{R}^{d}$ attached to each point, the random translation can be considered as the independent marking operation on a point process.

\subsection{Superposition}
Let $\Phi_1$ and $\Phi_2$ be two point processes. Consider the union
\begin{equation}
\label{eqn:Union_of_two_PPs}
\Phi = \Phi_1 \cup \Phi_2 \text{.}
\end{equation}
Suppose that with probability one the point sets $\Phi_1$ and $\Phi_2$ do not overlap. The set-theoretic union then coincides with the superposition operation of general point process theory. The superposition preserves the LF ordering of point processes as follows:
\begin{lemma}
\label{Lf_ordering_superposition_PP}
Let $\Phi_{1,i}$ and $\Phi_{2,i}, i=1,...,M$ be mutually independent point processes and $\Phi_1=\bigcup_{i=1}^M \Phi_{1,i}$ and $\Phi_2=\bigcup_{i=1}^M \Phi_{2,i}$ be the superposition of point processes. If $\Phi_{1,i} \leq_{\mathrm{Lf}} \Phi_{2,i}$ for $i=1,...,M$, then $\Phi_1 \leq_{\mathrm{Lf}} \Phi_2$.
\end{lemma}
\begin{proof}
The proof is given in Appendix \ref{app:Proof_for_superposition_PP}.
\end{proof}

\section{Applications to Wireless Networks}
In the following discussion, we will consider the applications of stochastic orders to wireless network systems.

\subsection{Cellular Networks}
\label{sec:App_to_Cellular_Networks}
In this section, the comparisons of performance metrics will be derived based on the LF ordering of point processes for spatial deployments of base stations (BSs) and mobile stations (MSs).

\subsubsection{System Model}
\label{subsec:System_Model}
\begin{figure}[tb]
\begin{minipage}{1\textwidth}
\centering
\begin{center}
\includegraphics[height=8.0cm,keepaspectratio]{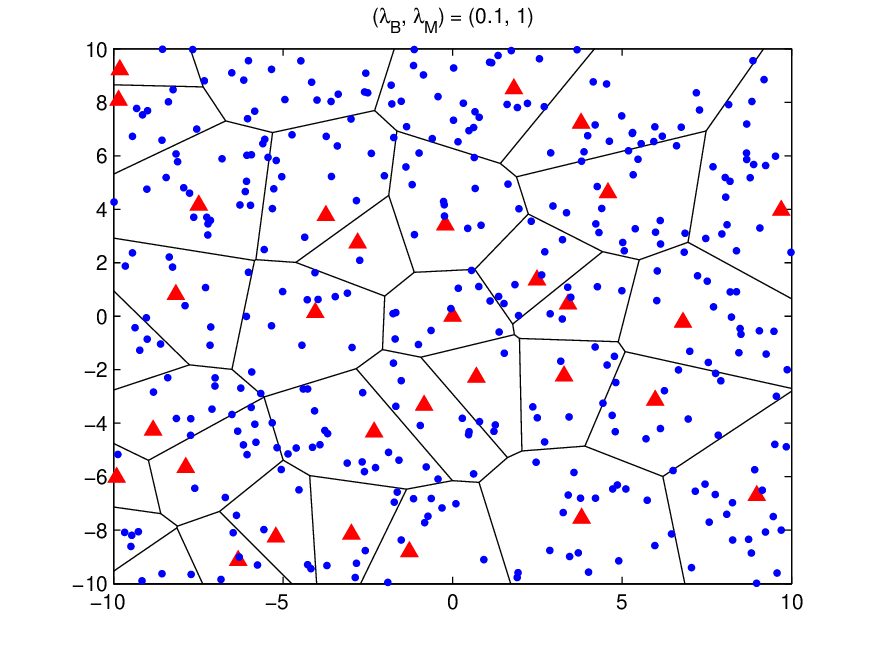}
\caption{Illustration of a cellular network. The triangles represent base stations which form a stationary Poisson point process $\Phi_{\mathrm{B}}$ with intensity $\lambda_{\mathrm{B}}=0.1$ and the dots represent users which also form a stationary Poisson point process $\Phi_{\mathrm{M}}$ with intensity $\lambda_{\mathrm{M}}=1.0$.}
\label{fig:System_model}
\end{center}
\end{minipage}
\end{figure}

We consider the downlink cellular network model consisting of BSs arranged according to some point process $\Phi_{_{\mathrm{B}}}$ in the Euclidean plane. For the deployment of BSs, a deterministic network such as lattice points or stochastic network such as a Poisson point process may be considered. Consider an independent collection of MSs, located according to some point process $\Phi_{_{\mathrm{M}}}$ which is independent of $\Phi_{_{\mathrm{B}}}$. Fig. \ref{fig:System_model} shows an example of cellular network consisting of stationary Poisson point processes with different intensities for BSs and MSs respectively. For a traditional cellular network, assume that each user associates with the closest BS, which would suffer the least path loss during wireless transmission. It is also assumed that the association between a BS and a MS is carried out in a large time scale compared to the coherence time of the channel. The cell boundaries are defined through the Voronoi tessellation of the BS process. Our goal is to compare performance metrics such as total cell coverage probability through stochastic ordering tools. The spatial coverage of cellular networks is also compared based on the LF order of the BS point processes.

In order for the total cell coverage probability to be compared, the signal to interference plus noise ratio (SINR) of a user at $x \in \Phi_{_{\mathrm{M}}}$ should be quantified. The effective channel power between a user $x$ and its associated typical BS $B_{0}$ is $h_{\mathrm{S}}^{(x)}$, which is a non-negative RV. The SINR with additive noise power $\sigma^2$ is given by
\begin{equation}
\label{eqn:def_SINR}
\emph{SINR}(x) = \frac{h_{\mathrm{S}}^{(x)}g(\Vert x \Vert)}{\sigma^2+I(x)} \text{,}
\end{equation}
where $g(\cdot):\mathbb{R}^{+}\rightarrow\mathbb{R}^{+}$ is the path-loss function which is a continuous, positive, non-increasing function of the Euclidean distance $\Vert x \Vert$ from the user located at $x$ to the typical BS $B_{0}$. The following is an example of a path-loss model \cite{Ilow1998, Haenggi2008, Win2009, Baccelli2009a}:
\begin{equation}
\label{eqn:path_loss_model}
g(\Vert x \Vert)=(a+b \Vert x \Vert^{\delta})^{-1}
\end{equation}
for some $b > 0$, $\delta > d$ and $a \in \lbrace 0,1 \rbrace$, where $\delta$ is called the path-loss exponent, $a$ determines whether the path-loss model belongs to a singular path-loss model ($a=0$) or a non-singular path-loss model ($a=1$), and $b$ is a compensation parameter to keep the total receive power normalized regardless the values of path-loss exponent. In \eqref{eqn:def_SINR}, $I(x)$ is the accumulated interference power at a user located at $x$ given by
\begin{equation}
\label{eqn:Interference_model}
I(x)=\sum_{y\in\Phi_{_{\mathrm{B}}}\setminus \{B_{0}\}}h_{\mathrm{I}}^{(y)} g(\Vert y-x \Vert)
\end{equation}
where $\Phi_{_{\mathrm{B}}}$ denotes the set of all BSs which is modeled as a point process and $h_{\mathrm{I}}^{(y)}$ is a positive random variable capturing the (power) fading coefficient between a user $x$ and the $y^{\mathrm{th}}$ interfering BS. Moreover, $h_{\mathrm{I}}^{(y)}$ are i.i.d. random variables and independent of $\Phi_{_{\mathrm{M}}}$ and $\Phi_{_{\mathrm{B}}}$.

\subsubsection{Ordering of Performance Metrics in a Cellular Network}
In the following discussion, we will introduce performance metrics involving the stochastic ordering of aggregate process in the cell $\mathfrak{C}_{0}$ which is associated with the BS $B_{0}$. By studying spatial character of networks and investigating the spatial distributions of mobile users, we can compare system performances using stochastic ordering approach without actual system performance evaluation. In addition, the preservation properties of the LF order in Lemma \ref{Lf_ordering_marked_PP}-\ref{Lf_ordering_superposition_PP} guarantee the performance comparison results based on the LF ordering of point processes are not changed with respect to any identical independent random operation on $\Phi_{_{\mathrm{M}}}$ or $\Phi_{_{\mathrm{B}}}$ such as marking, thinning, translation and superposition.

\paragraph{Total Cell Coverage Probability}
\label{subsec:Tot_Cell_Cov_Prob}
In multicast/broadcast scenarios, multiple users receive a common signal from their associated BS. Therefore, the probability that SINRs of all served users are greater than a minimum threshold $T$ is an important measure to ensure the signal reception quality of every user and it is called \emph{total cell coverage probability}. This metric can be ordered, if the underlying point processes are LF ordered:
\begin{theorem}
\label{Tot_Cell_Cov_Prob}
Let $\Phi_{_{\mathrm{B}}}$ be an arbitrary fixed BS deployment. Also, let $\Phi_{_{\mathrm{M}_1}}$ and $\Phi_{_{\mathrm{M}_2}}$ be two point processes for MS deployments, and $x_{1}^{1},..., x_{N_1}^{1} \in \Phi_{_{\mathrm{M}_1}} \cap \mathfrak{C}_0$ and $x_{1}^{2},..., x_{N_2}^{2} \in \Phi_{_{\mathrm{M}_2}} \cap \mathfrak{C}_{0}$ be two user location sets belong to the typical Voronoi cell $\mathfrak{C}_0$. $N_1$ and $N_2$ are random variables that are functions of $\Phi_{_{\mathrm{M}_1}}$ and $\Phi_{_{\mathrm{M}_2}}$ respectively. Let $\{h_{\mathrm{S}}^{(x)}\}, x \in (\Phi_{_{\mathrm{M}_1}}\cup\Phi_{_{\mathrm{M}_2}})\cap\mathfrak{C}_{0}$ in \eqref{eqn:def_SINR} be exponentially distributed independent RVs, and $\{h_{\mathrm{I}}^{(y)}\}, y \in \Phi_{_{\mathrm{B}}}$ in \eqref{eqn:Interference_model} be independent RVs, that are also independent of $\{h_{\mathrm{S}}^{(x)}\}$, and $\Phi_{_{\mathrm{M}_1}}$ and $\Phi_{_{\mathrm{M}_2}}$. If $\Phi_{_{\mathrm{M}_1}} \leq_{\mathrm{Lf}} \Phi_{_{\mathrm{M}_2}}$ then for $T>0$
\begin{equation}
\label{eqn:Tot_Cell_Cov_Prob}
P\left(\text{SINR}(x_{1}^{1}) \geq T,\dots,\text{SINR}(x_{N_1}^{1}) \geq T\right) \geq  P\left(\text{SINR}(x_{1}^{2}) \geq T,\dots,\text{SINR}(x_{N_2}^{2}) \geq T\right)
\end{equation}
where $P\left(\cdot\right)$ is over all RVs, $h_{\mathrm{S}}^{(x)}$, $h_{\mathrm{I}}^{(y)}$, $\Phi_{_{\mathrm{M}_1}}$, and $\Phi_{_{\mathrm{M}_2}}$.
\end{theorem}
\begin{proof}
The proof is given in Appendix \ref{app:Proof_for_Tot_Cell_Cov_Prob}.
\end{proof}
Note that Theorem \ref{Tot_Cell_Cov_Prob} holds for any arbitrary $\Phi_{_{\mathrm{B}}}$ as long as it is identical under both scenarios. Equation \eqref{eqn:Tot_Cell_Cov_Prob} shows that LF ordering of user point processes imply ordering of the multivariate complementary cumulative distribution functions (CCDFs) of the SINRs.

\paragraph{Network Spatial Coverage}
\label{subsec:Net_Spa_Cov}
The network spatial coverage is an important performance metric to design BS deployment in cellular networks. We assume that BSs are distributed by a point process $\Phi_{_{\mathrm{B}}}$ and each BS has a fixed radius of coverage $R$. Denote the random number of BSs covering a fixed location $y$ by
\begin{equation}
\label{eqn:net_spa_cov1}
S(y)=\sum_{x\in\Phi_{_{\mathrm{B}}}}\mathbf{1}\lbrace y \in B_x(R) \rbrace \text{,}
\end{equation}
where $B_x(R)$ is a $d$-dimensional ball of radius $R$ centered at the point $x$. Denote the probability generating function of the random number of BSs covering location $y$ by $G(t)=\mathbb{E}[t^{S(y)}], 0 \leq t \leq 1$. Since $0 \leq t \leq 1$, $t^z$ is a c.m. function with $z$. Note that $1-G(0)$ represents the probability whether the location $y$ is covered by at least one BS from the definition of the probability generation function. Thus, if $\Phi_{_{\mathrm{B}_1}} \leq_{\mathrm{Lf}} \Phi_{_{\mathrm{B}_2}}$, then $G_1(t) \geq G_2(t)$ from the property of LT ordering in \eqref{eqn:def_Lt_ordering_conseq1} and consequently $1-G_1(0) \leq 1-G_2(0)$. This means the probability that any arbitrary point $y$ in $\mathbb{R}^d$ is covered by at least one BS with the cell deployment by $\Phi_{_{\mathrm{B}_1}}$ is always less than the probability with $\Phi_{_{\mathrm{B}_2}}$. Due to random effects of real systems such as shadowing, different transmission power per each BS, and obstacles, the range of coverage $R$ can be a non-negative random variable. Since the random range of coverage $R$ can be considered as independent marking $m=R$ in Section \ref{subsec:independent_marking}, the ordering of spatial coverage probabilities still holds from Lemma \ref{Lf_ordering_marked_PP} under the assumption that the random range of coverage $R$ is independent of the BS deployment $\Phi_{_{\mathrm{B}}}$. The ordering of spatial coverage probabilities also holds with random ellipse or square instead of a ball $B_x(R)$ from Lemma \ref{Lf_ordering_marked_PP}.

With performance metrics such as network spatial coverage and network interference, our study can provide design guidelines for network deployment to increase spatial coverage of networks or provide less interference from networks. As an example, consider the effects of clustering as in Section \ref{subsec:Ordering_of_Cluseter_PPs}. In the case that the daughter clusters are assumed to be Poisson, the clustering of nodes diminishes a point process in the LF order. From an interference point of view, the clustering of interfering nodes causes less interference in the LT order between interference distributions. This translates into an increased coverage probability and improved capacity for the system. However, the clustering of nodes also causes less spatial coverage. Therefore, the proper point process for network deployment should be studied for balancing between interference and spatial coverage.

\subsection{Cognitive Networks}
\label{sec:App_to_Cognitive_Networks}
In the following discussion, we will consider the applications of stochastic orders to cognitive network systems where there is an increasing interest in developing efficient methods for spectrum management and sharing.

\subsubsection{System Model}
\begin{figure}[tb]
\begin{minipage}{1\textwidth}
\centering
\begin{center}
\includegraphics[height=8.0cm,keepaspectratio]{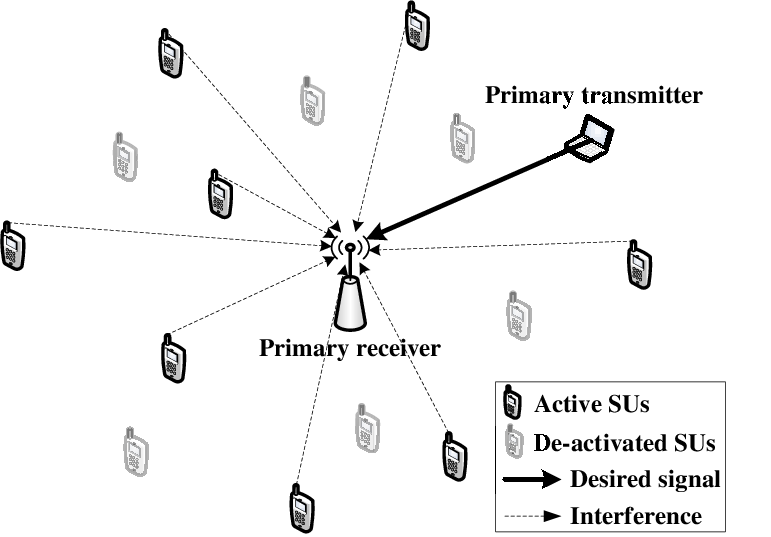}
\caption{Illustration of a cognitive network.}
\label{fig:Cognitive_network_SUs}
\end{center}
\end{minipage}
\end{figure}

Let us consider an underlay cognitive radio network which contains a primary user (PU) and many secondary users (SUs) with an average interference power constraint $\Gamma_{\mathrm{I}}$. The PU is located at the origin. The $L$ SUs are uniformly randomly located in a certain area $B \subset \mathbb{R}^d$ \cite{Pai2009, Sung2011, Grace2012}. It is assumed that there is a BS to coordinate the SUs' transmission. In order to satisfy the interference constraint, the BS selects only $N$ active number of users out of $L$ total users through user selection schemes and the selected active SUs are allowed to transmit their signals to the BS. User selection creates a thinned point process $\Phi_{_{\mathrm{SU,th}}}$ for the active SUs which has uniformly distributed $N$ points over the area $B$. The resulting thinned point process can be considered as a mixed binomial point processes as in Section \ref{subsec:Mixed_BPPs}. The system model of the cognitive radio network is illustrated in Fig. \ref{fig:Cognitive_network_SUs}. Under the system model, based on Theorem \ref{Lf_ordering_of_mixed_BPPs}, one can design user selection schemes which guarantee the same average interference power from the active SUs to the PU and the same average sum rate for the active SUs. On the other hand, the \emph{distribution} of the instantaneous interferences from the active SUs to the PU are such that they are LT ordered and cause ordered performance metrics such as a coverage probability and an achievable rate of the PU according to the user selection schemes. We now discuss these in detail.

\paragraph{Average Interference Power Constraint}
\label{avg_interf_from_SUs}
The instantaneous aggregate interference from the active SUs to the PU, can be expressed as:
\begin{equation}
\label{eqn:def_agg_interf_from_SUs}
I_{_{\mathrm{SU}}}=\sum_{x\in\Phi_{_{\mathrm{SU,th}}}}h_{\mathrm{I}}^{(x)}g(\Vert x \Vert) \text{,}
\end{equation}
where $\Phi_{_{\mathrm{SU,th}}}$ is a point process for the active SUs and $h_{\mathrm{I}}^{(x)}$ is a positive random variable capturing the (power) fading coefficient between an active SU located at $x$ and the PU. From Campbell's theorem, the average of aggregate interference from the active SUs, $\mathbb{E}[I_{_{\mathrm{SU}}}]$ is the same as long as the average number of the active SUs is fixed to $\mu=\mathbb{E}[N]$ regardless of the distribution for $N$ (For a proof, please see Appendix \ref{app:Proof_for_Cognitive_networks}). Therefore, we need to select $\mu$ in order to satisfy the average interference power constraint $\mathbb{E}[I_{_{\mathrm{SU}}}] \leq \Gamma_{\mathrm{I}}$.

\paragraph{Average Sum Rate}
\label{avg_sum_of_rates_of_SUs}
On the other hand, if there is no interference between the active SUs by adopting code-division multiple access (CDMA), the instantaneous sum of achievable rates of the active SUs, $C_{_{\mathrm{SU}}}$ is also random and can be expressed as
\begin{equation}
\label{eqn:sum_of_rates_of_SUs}
C_{_{\mathrm{SU}}}(z)=\sum_{x\in\Phi_{_{\mathrm{SU,th}}}}\log\left(1+\frac{h_{\mathrm{S}}^{(x)}g(\Vert x-z \Vert)}{\sigma^2+I_{_{\mathrm{P}}}(x)}\right) \text{,}
\end{equation}
where $h_{\mathrm{S}}^{(x)}$ is a effective fading channel between an active SU located at $x$ and the BS for SUs located at $z$, $I_{_{\mathrm{P}}}(x)=h_{\mathrm{I}}^{(x)}g(\Vert x \Vert)$ is the interference power from the PU to the SU located at $x$, and $\sigma^2$ is a additive noise power. From Campbell's theorem, the average of sum of achievable rates of active SUs, $\mathbb{E}[C_{_{\mathrm{SU}}}(z)]$ is the same regardless of the point processes for the active SUs as long as the average number of the active SUs is equal to $\mu=\mathbb{E}[N]$ (Proof is the same as the average interference case in Appendix \ref{app:Proof_for_Cognitive_networks} with a different function).

\subsubsection{Randomized User Selection Scheme for Secondary Users}
\label{subsubsec:User_Selection_Scheme}
We now discuss different ways to perform randomized user selection which corresponds to different ways of thinning the SU point process. One can apply the stochastic ordering tool to design user selection schemes based on Theorem \ref{Lf_ordering_of_mixed_BPPs} which produce mixed binomial point processes for active SUs. Given $\mu=\mathbb{E}[N]$ which satisfies the interference constraint $\mathbb{E}[I_{_{\mathrm{SU}}}]\leq\Gamma_{\mathrm{I}}$, the smaller LT ordered random number of active SUs $N$ provides the smaller LF ordered point process $\Phi_{_{\mathrm{SU,th}}}$, that is $N_1 \leq_{\mathrm{Lt}} N_2 \Rightarrow \Phi_{_{\mathrm{SU,th_1}}} \leq_{\mathrm{Lf}} \Phi_{_{\mathrm{SU,th_2}}}$ from Theorem \ref{Lf_ordering_of_mixed_BPPs}. Consequently, the resulting mixed binomial point processes cause LT ordered aggregate interferences to the PU, $I_{_{\mathrm{SU_1}}} \leq_{\mathrm{Lt}} I_{_{\mathrm{SU_2}}}$. The LT ordered aggregate interferences from the active SUs to the PU yield ordered performance metrics for the PU such as a coverage probability and an achievable rate due to the LT ordered interferences and the coverage probability and the achievable rate having c.m. property with respect to the interferences \cite{Lee2014}. On the other hand, the average sum of achievable rates of the active SUs and the average of interference power remain the same, $\mathbb{E}[C_{_{\mathrm{SU}_1}}] = \mathbb{E}[C_{_{\mathrm{SU}_2}}]$ and $\mathbb{E}[I_{_{\mathrm{SU}_1}}] = \mathbb{E}[I_{_{\mathrm{SU}_2}}]$ as discussed in Section \ref{avg_interf_from_SUs} and \ref{avg_sum_of_rates_of_SUs}.

We now give examples of a user selection scheme using the stochastic ordering approach. If the active SUs are chosen among $L$ total SUs with a probability $p$ independently, the number of active SUs is a binomial random variable, $N_{_{\mathrm{B}}}$. The number of active SUs can follow discrete distributions other than binomial if different modes of operation are adopted. For another example, if the SUs are selected before a predetermined failure numbers $r$ with a probability $p$ occurs, then the number of selected active SUs follows a negative binomial distribution with parameter $r$ and $p$ denoted as $N_{_{\mathrm{NB}}}$. Since $N_{_{\mathrm{NB}}} \leq_{\mathrm{Lt}} N_{_{\mathrm{B}}}$ \cite{Zeng2014}, $\Phi_{_{\mathrm{NB}}} \leq_{\mathrm{Lf}} \Phi_{_{\mathrm{B}}}$ from Theorem \ref{Lf_ordering_of_mixed_BPPs}. Therefore, the aggregate interferences from the active SUs are LT ordered $I_{_{\mathrm{NB}}} \leq_{\mathrm{Lt}} I_{_{\mathrm{B}}}$, while $\mathbb{E}[C_{_{\mathrm{NB}}}] = \mathbb{E}[C_{_{\mathrm{B}}}]$ and $\mathbb{E}[I_{_{\mathrm{NB}}}] = \mathbb{E}[I_{_{\mathrm{B}}}] \leq \Gamma_{\mathrm{I}}$. For these operations, the BS for SUs only needs to know the number of total SUs $L$ and the average number of active SUs $\mu$.

From the discussion in previous Section \ref{avg_interf_from_SUs} and \ref{avg_sum_of_rates_of_SUs}, it is noted that the smaller LT ordered random number of active SUs $N$ causes less interference to the PU, while the average interference $\mathbb{E}[I_{_{\mathrm{SU}}}]$ and average sum rate $\mathbb{E}[C_{_{\mathrm{SU}}}]$ are the same. From Theorem 5.A.21 in \cite{Shaked} and the discussion in Section \ref{subsec:Pert_lat_PPs}, the smallest LT ordered random number $N_{\mathrm{min}}$ is $P\lbrace N=0 \rbrace=1-P\lbrace N=2\mu \rbrace=1/2$ and the biggest $N_{\mathrm{max}}$ is the fixed $N=\mu$. Even though the $N_{\mathrm{min}}$ causes the smallest LT ordered interference to the PU, when $N=2\mu$, it causes large instantaneous interference to the PU at this moment. Otherwise, the $N_{\mathrm{max}}$ provides balanced traffic from the active SUs since the fixed number of random set of users $\mu$ is always selected. However, it causes the bigger LT ordered interference to the PU than any other distributions for active random SUs $N$. Therefore, the proper distribution for random number of active SUs should studied for balancing peakyness of instantaneous interference power and fairness of active SUs' traffic.

\section{Numerical Results}
\label{sec:Numerical_Results}
In this section, we verify our theoretical results through Monte Carlo simulations.

\begin{figure}[tb]
\begin{minipage}{1\textwidth}
\centering
\begin{center}
\includegraphics[height=8.0cm,keepaspectratio]{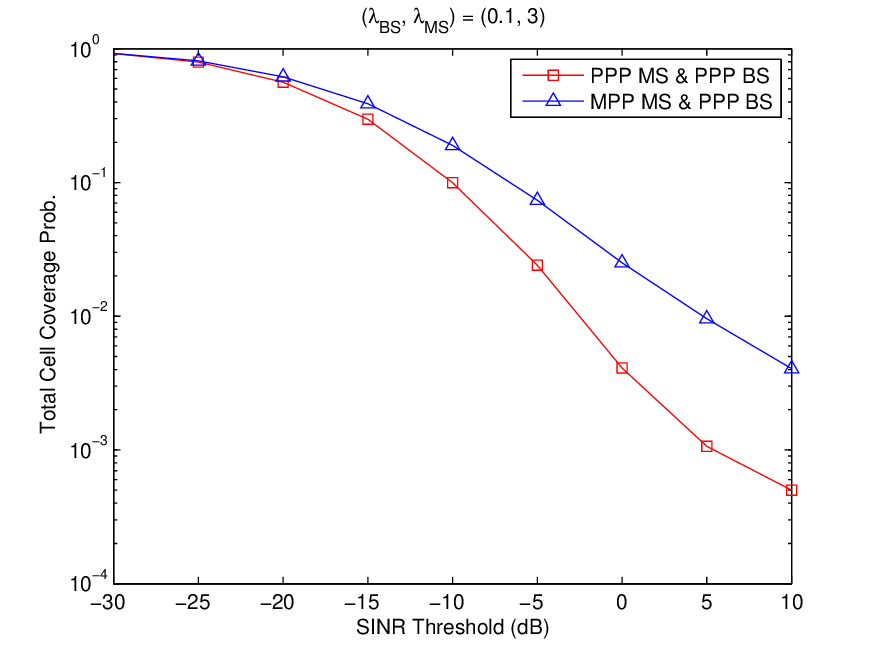}
\caption{Total cell coverage probabilities}
\label{fig:Tot_cell_cov_prob_PPP}
\end{center}
\end{minipage}
\end{figure}

\subsection{Cellular Networks}
We show in Fig. \ref{fig:Tot_cell_cov_prob_PPP} the total cell coverage probabilities. It is assumed that the BS distribution is a stationary Poisson point process $\Phi_{_{\mathrm{B}}}$, and the compared user distributions follow also a Poisson point process $\Phi_{_{\mathrm{PPP}}}$ and mixed Poisson process $\Phi_{_{\mathrm{MPP}}}$ with same intensity $\lambda_{\mathrm{M}}$. Since $\Phi_{_{\mathrm{ MPP}}} \leq_{\mathrm{Lf}} \Phi_{_{\mathrm{PPP}}}$, from Theorem \ref{Tot_Cell_Cov_Prob}, the total cell coverage probability of the users distributed by $\Phi_{_{\mathrm{MPP}}}$ is greater than that of $\Phi_{_{\mathrm{PPP}}}$. Using the stochastic ordering approach, one can compare the coverage probabilities by investigating the spatial distributions of mobile users without actual system evaluation.

\subsection{Cognitive Networks}
\begin{figure}[tb]
\begin{minipage}{1\textwidth}
\centering
\begin{center}
\includegraphics[height=8.0cm,keepaspectratio]{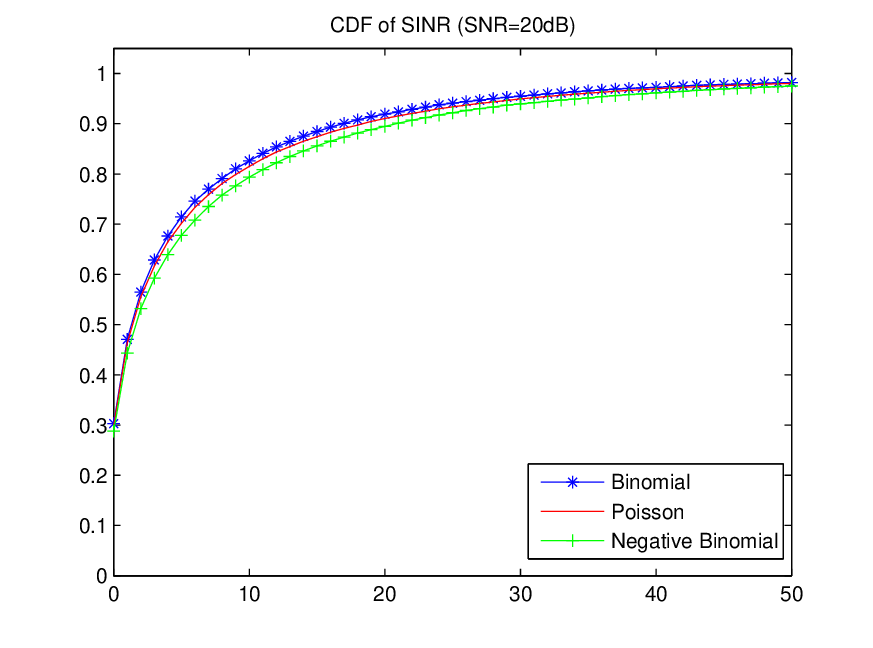}
\caption{Coverage probabilities of primary user}
\label{fig:CDF_SIR_cognitive_networks}
\end{center}
\end{minipage}
\end{figure}

\begin{figure}[tb]
\begin{minipage}{1\textwidth}
\centering
\begin{center}
\includegraphics[height=8.0cm,keepaspectratio]{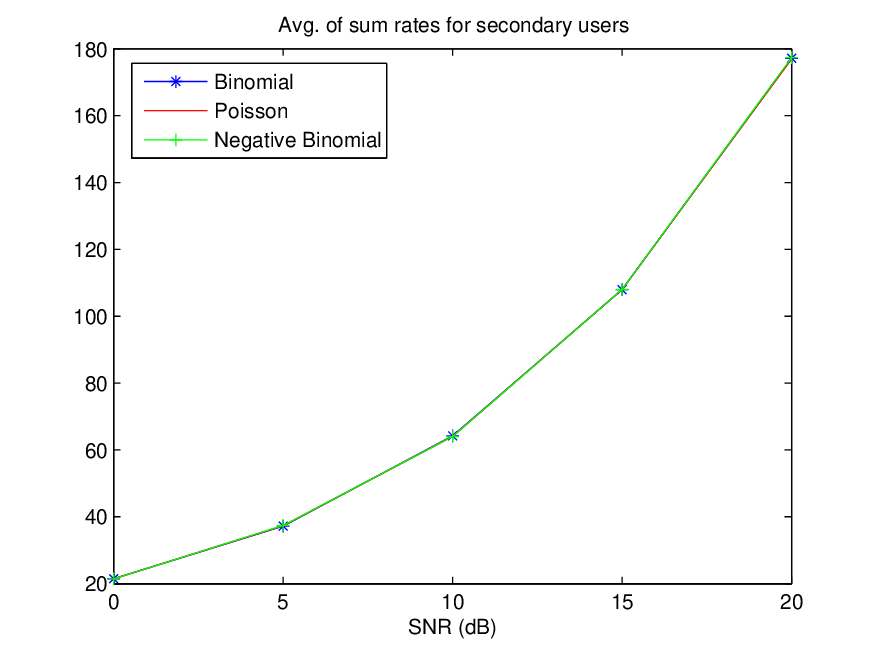}
\caption{Average sum rates of secondary users}
\label{fig:Avg_rates_cognitive_networks}
\end{center}
\end{minipage}
\end{figure}

In Fig. \ref{fig:CDF_SIR_cognitive_networks}, the CDFs of the SINR of the primary user are shown where the spatial distributions of the active SUs are different mixed binomial processes. The LT ordered number of the active SUs $N_{_{\mathrm{NB}}} \leq_{\mathrm{Lt}} N_{_{\mathrm{P}}} \leq_{\mathrm{Lt}} N_{_{\mathrm{B}}}$ of the negative binomial, Poisson and binomial RVs ensures the ordering $\Phi_{_{\mathrm{NB}}} \leq_{\mathrm{Lf}} \Phi_{_{\mathrm{P}}} \leq_{\mathrm{Lf}} \Phi_{_{\mathrm{B}}}$ of the corresponding mixed binomial processes, from Theorem \ref{Lf_ordering_of_mixed_BPPs}. Since the aggregate interferences from LF ordered point processes for spatial distributions of the active SUs are LT ordered, we observe $\emph{SIR}_{_{\mathrm{NB}}} \geq_{\mathrm{st}} \emph{SIR}_{_{\mathrm{P}}} \geq_{\mathrm{st}} \emph{SIR}_{_{\mathrm{B}}}$ which are the ratio between the PU's effective fading channel and the interference from active SUs \cite{Lee2014}. This is because when the PU's effective fading channel is exponentially distributed and interferences are LT ordered, the SIRs are reverse ordered. The aforementioned relationship between SIR and interference distributions holds as long as the CCDF of the PU's effective fading channel is a c.m. function by the property of LT ordering in \eqref{eqn:def_Lt_ordering_conseq1}. An example is the exponential distribution. Despite the ordering of the SIRs in Fig. \ref{fig:CDF_SIR_cognitive_networks}, the average of the sum of achievable rates of the active SUs in \eqref{eqn:sum_of_rates_of_SUs} are the same regardless of LF ordered point processes as shown in Fig. \ref{fig:Avg_rates_cognitive_networks}. In order to cause less interference from active SUs to the PU, one can use the stochastic approach to design the user selection schemes with the minimum information such as the total number of SUs $L$ and the average number of active SUs $\mu$ as discussed in Section \ref{sec:App_to_Cognitive_Networks}.

\section{Summary}
\label{sec:Summary}
In this paper, Laplace functional ordering of broad classes of point processes are investigated. We showed that the preservation of LF ordering of point process with respect to several operations, such as independent marking, thinning, random translation and superposition. We introduced the applications of LF ordering of point processes to wireless networks such as cellular networks and cognitive networks, which provided guidelines for design of user selection schemes and transmission strategy for wireless networks. Tradeoffs between coverage and interference as well as fairness and peakyness were also discussed. The power of this approach is that network performance comparisons can be made even in cases where a closed form expression for the performances is not analytically tractable. We verified our results through Monte Carlo simulations.

\section*{Conflicts of Interest}
The authors declare that there is no conflict of interest regarding the publication of this paper.

\appendix
\section{}
\label{app:Proofs_for_All}

\subsection{Proof of Theorem \ref{Lf_ordering_of_Cox_PPs}}
\label{app:Proof_for_Lf_Cox_PPs}
A Cox process is a Poisson point process conditional on the realization of the intensity measure. Therefore, the Laplace functional of the Cox process $\Phi$ with driving random measure $\Psi$ can be expressed as follows \cite{Stoyan1995}:
\begin{eqnarray}
\label{eqn:app_Lf_Cox_PPs1}
L_{\Phi}(u) &=& \int_{\mathbb{M}}L_{\Lambda}(u)\Psi(\mathrm{d}\Lambda) \\
\label{eqn:app_Lf_Cox_PPs2}
&=& \int_{\mathbb{M}}\exp \left(-\int_{\mathbb{R}^d}[1-\exp(-u(x))]\Lambda(\mathrm{d}x)\right)\Psi(\mathrm{d}\Lambda) \\
\label{eqn:app_Lf_Cox_PPs3}
&=& L_{\Psi}(1-\exp(-u)) \text{,}
\end{eqnarray}
where $L_{\Lambda}$ is the Laplace functional of the Poisson process of intensity measure $\Lambda$ and $\mathbb{M}$ is the set of intensity measures $\Lambda$ in \eqref{eqn:app_Lf_Cox_PPs1}. Equation \eqref{eqn:app_Lf_Cox_PPs2} follows from the definition of Laplace functional of Poisson point process in \eqref{eqn:def_Lf_of_PP} and \eqref{eqn:app_Lf_Cox_PPs3} follows from the definition of Laplace functional of random measure in \eqref{eqn:def_Lf_ordering}. Then, the proof follows from Definition \ref{def_Lf_ordering}.

\subsection{Proof of Theorem \ref{Lf_ordering_of_homo_indep_cluster_PPs}}
\label{app:Proof_for_homo_indep_cluster_PPs}
Let $L_0^{(x)}$ denote the Laplace functional of $\Phi_0+x$, i.e., the Laplace functional of the representative cluster translated by $x$. It can be expressed as follows:
\begin{equation}
\label{eqn:Lf_homo_indep_cluster_PPs1}
L_0^{(x)}(u) := \mathbb{E}_{\Phi_0}\left[\prod_{y\in\Phi_0+x}\exp(-u(y))\right] = \mathbb{E}_{\Phi_0}\left[\prod_{y\in\Phi_0}\exp(-u(y+x))\right] \text{.}
\end{equation}
From \eqref{eqn:Lf_homo_indep_cluster_PPs1}, the Laplace functional of the homogeneous independent cluster process $\Phi$ with parent process $\Phi_{\mathrm{p}}$ can be expressed as follows \cite{Stoyan1995}:
\begin{equation}
\label{eqn:Lf_homo_indep_cluster_PPs2}
L_{\Phi}(u) = \mathbb{E}_{\Phi_{\mathrm{p}}}\left[\prod_{x\in\Phi_{\mathrm{p}}}L_0^{(x)}(u)\right] \text{, } u \in \mathscr{U} \text{.}
\end{equation}
From Definition \ref{def_Lf_ordering} and \eqref{eqn:Lf_homo_indep_cluster_PPs2}, if $\Phi_{\mathrm{p}_1} \leq_{\mathrm{Lf}} \Phi_{\mathrm{p}_2}$ and same $L_0^{(x)}(u)$ in both, then $\Phi_1\leq_{\mathrm{Lf}} \Phi_2$. Similarly, if $\Phi_{0_1} \leq_{\mathrm{Lf}} \Phi_{0_2}$ and same $\Phi_{\mathrm{p}}$ in both, then $\Phi_1\leq_{\mathrm{Lf}} \Phi_2$. Therefore, if $\Phi_{\mathrm{p}_1} \leq_{\mathrm{Lf}} \Phi_{\mathrm{p}_2}$ and $\Phi_{0_1} \leq_{\mathrm{Lf}} \Phi_{0_2}$, then $\Phi_1\leq_{\mathrm{Lf}} \Phi_2$.

\subsection{Proof of Theorem \ref{Lf_ordering_of_pert_lat_PPs}}
\label{app:Proof_for_Lf_pert_lat_PPs}
The uniformly perturbed lattice process can be considered as the superpositions of independent finite point processes corresponding to each Voronoi cell. The finite point process in one of Voronoi cells consists of uniformly distributed random number of points $X$. The Laplace functional of the finite point process can be expressed as follows:
\begin{equation}
\label{eqn:app_Lf_pert_lat1}
L_{\Phi}(u)=\mathbb{E}_{X}\left[\left(\frac{1}{\lambda|V|}\int_{V}\exp(-u(x))\lambda\mathrm{d}x\right)^{X}\right] \text{,}
\end{equation}
where $V$ is a region of the Voronoi cell. By denoting $t = 1/(\lambda|V|)\int_{V}\exp(-u(x))\lambda\mathrm{d}x$, $t^z$ is a c.m. function of $z$ since $1/(\lambda|V|)\int_{V}\exp(-u(x))\lambda\mathrm{d}x \leq 1$. Therefore, if $X_1 \leq_{\mathrm{Lt}} X_2$, then we have the following relation:
\begin{equation}
\label{eqn:app_Lf_pert_lat2}
\mathbb{E}_{X_1}\left[t^{X_1}\right] \geq \mathbb{E}_{X_2}\left[t^{X_2}\right] \text{.}
\end{equation}
From \eqref{eqn:app_Lf_pert_lat2} the following LF ordering is obtained,
\begin{equation}
\label{eqn:app_Lf_pert_lat3}
\mathbb{E}_{X_1}\left[\left(\frac{1}{\lambda|V|}\int_{V}\exp(-u(x))\lambda\mathrm{d}x\right)^{X_1}\right] \geq \mathbb{E}_{X_2}\left[\left(\frac{1}{\lambda|V|}\int_{V}\exp(-u(x))\lambda\mathrm{d}x\right)^{X_2}\right] \text{.}
\end{equation}
Since the finite point processes are LF ordered, their superpositions are also LF ordered from Lemma \ref{Lf_ordering_superposition_PP}. The proof of Theorem \ref{Lf_ordering_of_pert_lat_PPs} is completed.

\subsection{Proof of Theorem \ref{Lf_ordering_of_mixed_BPPs}}
\label{app:Proof_for_Lf_mixed_BPPs}
Similar to Appendix \ref{app:Proof_for_Lf_pert_lat_PPs}, the Laplace functional of a mixed binomial point process with random number of points $N$ can be expressed as follows:
\begin{equation}
\label{eqn:app_Lf_mixed_BPPs1}
L_{\Phi}(u)=\mathbb{E}_{N}\left[\left(\frac{1}{\lambda|B|}\int_{B}\exp(-u(x))\lambda\mathrm{d}x\right)^{N}\right] \text{,}
\end{equation}
where $B$ is the bounded set of the point process. By denoting $t = 1/(\lambda|B|)\int_{B}\exp(-u(x))\lambda\mathrm{d}x$, $t^z$ is a c.m. function of $z$ since $1/(\lambda|B|)\int_{B}\exp(-u(x))\lambda\mathrm{d}x \leq 1$. Therefore, if $N_1 \leq_{\mathrm{Lt}} N_2$, then we have the following relation:
\begin{equation}
\label{eqn:app_Lf_mixed_BPPs2}
\mathbb{E}_{N_1}\left[t^{N_1}\right] \geq \mathbb{E}_{N_2}\left[t^{N_2}\right] \text{.}
\end{equation}
From \eqref{eqn:app_Lf_mixed_BPPs2} the following LF ordering is obtained,
\begin{equation}
\label{eqn:app_Lf_mixed_BPPs3}
\mathbb{E}_{N_1}\left[\left(\frac{1}{\lambda|B|}\int_{B}\exp(-u(x))\lambda\mathrm{d}x\right)^{N_1}\right] \geq \mathbb{E}_{N_2}\left[\left(\frac{1}{\lambda|B|}\int_{B}\exp(-u(x))\lambda\mathrm{d}x\right)^{N_2}\right] \text{.}
\end{equation}
The proof of Theorem \ref{Lf_ordering_of_mixed_BPPs} is completed.

\subsection{Proof of Lemma \ref{Lf_ordering_marked_PP}}
\label{app:Proof_for_marked_PPs}
The Laplace functional of an independently marked point process $\widetilde{\Phi}$ with a non-negative function $u:\mathbb{R}^d\times\mathbb{R}^{\ell}\mapsto\mathbb{R}^{+}$ can be expressed as follows \cite{Daley2003}:
\begin{eqnarray}
\label{eqn:Lf_marked_PPs1}
L_{\widetilde{\Phi}}(u)&=&\mathbb{E}_{\Phi}\left[\prod_{x\in\Phi}\int_{\mathbb{R}^{\ell}}\exp(-u(m,x))F(\mathrm{d}m \vert x)\right] \\
\label{eqn:Lf_marked_PPs2}
&=&\mathbb{E}_{\Phi}\left[\prod_{x\in\Phi}\exp\left(-\left(-\log\int_{\mathbb{R}^{\ell}}\exp\left(-u(m, x)\right)F(\mathrm{d}m \vert x)\right)\right)\right] \\
\label{eqn:Lf_marked_PPs3}
&=&\mathbb{E}_{\Phi}\left[\prod_{x\in\Phi}\exp\left(-\tilde{u}(x)\right)\right]
=\mathbb{E}_{\Phi}\left[\exp\left(-\sum_{x\in\Phi}\tilde{u}(x)\right)\right] \text{.}
\end{eqnarray}
where $\tilde{u}(x)=-\log\int_{\mathbb{R}^{\ell}}\exp\left(-u(m,x)\right)F(\mathrm{d}m \vert x)$. Since $0 \leq \int_{\mathbb{R}^{\ell}}\exp\left(-u(m,x)\right)F(\mathrm{d}m \vert x) \leq 1$, $\tilde{u}(x)$ is a non-negative function of $x$. Then, the Laplace functional of marked point process $\widetilde{\Phi}$ follows from
\begin{equation}
\label{eqn:Lf_marked_PPs4}
L_{\widetilde{\Phi}}(u)=L_{\Phi}(\tilde{u}) \text{.}
\end{equation}
Therefore, from \eqref{eqn:Lf_marked_PPs4} and the definition of LF ordering in \eqref{eqn:def_Lf_ordering}, if $\Phi_1 \leq_{\mathrm{Lf}} \Phi_2$, then $\widetilde{\Phi}_1 \leq_{\mathrm{Lf}} \widetilde{\Phi}_2$.

\subsection{Proof of Lemma \ref{Lf_ordering_of_thinning_PPs}}
\label{app:Proof_for_thinning_PPs}
If $L_{\Phi}$ is the Laplace functional of $\Phi$ then that of $\Phi_{\mathrm{th}}$ is
\begin{equation}
L_{\Phi_{\mathrm{th}}}(u)=L_{\Phi}(u_{\mathrm{p}}) \text{ for } u\in\mathscr{U} \text{,}
\end{equation}
where $u_{\mathrm{p}}(x)=-\log\left(\exp(-u(x))p(x)+1-p(x)\right)$. From Definition \ref{def_Lf_ordering}, if $\Phi_1 \leq_{\mathrm{Lf}} \Phi_2$, then $\Phi_{\mathrm{th},1} \leq_{\mathrm{Lf}} \Phi_{\mathrm{th},2}$. The $p$-thinning is subset of $p(x)$-thinning. Analogous formula for $\pi(x)$-thinning follows by averaging with respect to the distribution of the random process $\boldsymbol{\pi}$. Since the inequality holds under every realization $\pi(x)$, their expectations also hold the inequality.

\subsection{Proof of Lemma \ref{Lf_ordering_translation_PP}}
\label{app:Proof_for_translation_PPs}
Let $F(\cdot)$ denote the common distribution for the translations $t$. For $u\in\mathscr{U}$, the Laplace functional after random translation takes the form
\begin{equation}
L_{\Phi_{\mathrm{rt}}}(u)=L_{\Phi}\left(u_{\mathrm{t}}\right) \text{,}
\end{equation}
where $u_{\mathrm{t}}(x)=-\log\left(\int_{\mathbb{R}^d}\exp(-u(x+t))F(\mathrm{d}t\vert x)\right)$. From Definition \ref{def_Lf_ordering}, if $\Phi_1 \leq_{\mathrm{Lf}} \Phi_2$, then $\Phi_{\mathrm{rt},1} \leq_{\mathrm{Lf}} \Phi_{\mathrm{rt},2}$.

\subsection{Proof of Lemma \ref{Lf_ordering_superposition_PP}}
\label{app:Proof_for_superposition_PP}
Let $\Phi_{1,i}$ and $\Phi_{2,i}, i=1,...,M$ be mutually independent point processes and $\Phi_1=\bigcup_{i=1}^M \Phi_{1,i}$ and $\Phi_2=\bigcup_{i=1}^M \Phi_{2,i}$ be the superposition of point processes. The Laplace function of superposition of mutually independent point processes can be expressed as follows:
\begin{equation}
\label{eqn:Lf_sum_pp}
L_{\Phi}(u)=\prod_{i=1}^{M}L_{\Phi_i}(u) \text{.}
\end{equation}
$L_{\Phi}(u)$ converges if and only if the infinite sum of point processes is finite on bounded area $B\in\mathbb{R}^{d}$. Therefore, from \eqref{eqn:Lf_sum_pp} and the definition of LF ordering in \eqref{eqn:def_Lf_ordering}, if $\Phi_{1,i} \leq_{\mathrm{Lf}} \Phi_{2,i}$ for $i=1,...,M$, then $\Phi_1 \leq_{\mathrm{Lf}} \Phi_2$.

\subsection{Proof of Theorem \ref{Tot_Cell_Cov_Prob}}
\label{app:Proof_for_Tot_Cell_Cov_Prob}
To prove Theorem \ref{Tot_Cell_Cov_Prob}, we first express for a generic $\Phi_{_{\mathrm{M}}}$,
\begin{eqnarray}
\label{eqn:tot_cov_prob1}
\lefteqn{P\left(\emph{\text{SINR}}(x_{1}) \geq T,\dots,\emph{\text{SINR}}(x_{N}) \geq T\right)} \\
\label{eqn:tot_cov_prob2}
&=& \mathbb{E}_{\Phi_{_{\mathrm{M}}}}\left[P\left(\emph{\text{SINR}}(x_{1}) \geq T,\dots,\emph{\text{SINR}}(x_{N}) \geq T\right) \Big\vert \Phi_{_{\mathrm{M}}}\right] \\
\label{eqn:tot_cov_prob3}
&=& \mathbb{E}_{\Phi_{_{\mathrm{M}}}}\left[\prod_{x \in C_{0}\cap\Phi_{_{\mathrm{M}}}} P\left(\emph{\text{SINR}}(x) \geq T\right) \Big\vert \Phi_{_{\mathrm{M}}}\right] \\
\label{eqn:tot_cov_prob4}
&=& \mathbb{E}_{\Phi_{_{\mathrm{M}}}}\left[\prod_{x \in C_{0}\cap\Phi_{_{\mathrm{M}}}} P\left(h_{\mathrm{S}}^{(x)} \geq \frac{T(\sigma^2+I(x))}{g(\Vert x \Vert)}\right) \Big\vert \Phi_{_{\mathrm{M}}}\right] \\
\label{eqn:tot_cov_prob5}
&=& \mathbb{E}_{\Phi_{_{\mathrm{M}}}}\left[\exp\left(-\sum_{x \in \Phi_{_{\mathrm{M}}}}\frac{T(\sigma^2+I(x))}{g(\Vert x \Vert)}\mathbf{1}\{x\in C_{0}\}\right) \Big\vert \Phi_{_{\mathrm{M}}}\right] \text{.}
\end{eqnarray}
Equation \eqref{eqn:tot_cov_prob3} follows from the assumption that $h_{\mathrm{S}}$ and $h_{\mathrm{I}}$ in \eqref{eqn:def_SINR} and \eqref{eqn:Interference_model} are independent under the given realizations of $\Phi_{_{\mathrm{M}}}$. From the assumption that $h_{\mathrm{S}}^{(x)}$ is exponential distributed (Rayleigh fading), \eqref{eqn:tot_cov_prob5} follows. In \eqref{eqn:tot_cov_prob5}, $u(x)=(T(\sigma^2+I(x))/g(\Vert x \Vert))\mathbf{1}\{x\in C_{0}\}$ is a non-negative function of $x$. Therefore, if $\Phi_{_{\mathrm{M}_1}} \leq_{\mathrm{Lf}} \Phi_{_{\mathrm{M}_2}}$, then \eqref{eqn:Tot_Cell_Cov_Prob} follows by \eqref{eqn:def_Lt_ordering_conseq1} because $\exp(-z)$ is a c.m. function.

\subsection{Proof of $\mathbb{E}[I_{\mathrm{SU}}]$}
\label{app:Proof_for_Cognitive_networks}
From Campbell's theorem in \eqref{eqn:Campbell_Theorem_PP} and the intensity measure of the mixed binomial point processes $\Lambda(B)=\lambda|B|$, the average of interference can be expresses as follows:
\begin{eqnarray}
\label{eqn:app_Cogn_net1}
\mathbb{E}[I_{_{\mathrm{SU}}}] &=&\mathbb{E}_{\Phi_{_\mathrm{SU,th}},h_{\mathrm{I}}}\left[\sum_{x\in\Phi_{_\mathrm{SU,th}}} h_{\mathrm{I}}^{(x)}g(\Vert x \Vert) \Big\vert \Phi_{_\mathrm{SU,th}},h_{\mathrm{I}}\right] \\
\label{eqn:app_Cogn_net2}
&=&\mathbb{E}_{h_{\mathrm{I}}}\left[\underbrace{\lambda\int_{B}h_{\mathrm{I}}^{(x)}g(\Vert x \Vert)\mathrm{d}x}_{A} \Big\vert h_{\mathrm{I}}\right] \text{.}
\end{eqnarray}
Conditioned on $h_{\mathrm{I}}$, $A$ is always equal. Therefore, the expectation of \eqref{eqn:app_Cogn_net2} with respect to $h_{\mathrm{I}}$ is still equal. The proof is completed. 

\bibliographystyle{IEEEtran}
\nocite{*}
\bibliography{references}

\end{document}